\pdfoutput=1
\documentclass[journal]{IEEEtran}   
\usepackage{multirow}
\usepackage{diagbox}
\usepackage{amssymb}
\usepackage{amsmath}
\usepackage{graphicx}
\usepackage{cite}
\usepackage{citesort}
\usepackage{booktabs}
\usepackage{subfigure}
\usepackage{graphicx,epstopdf}
\usepackage{epsfig}	
\usepackage{cite,graphicx,amsmath,amssymb}
\usepackage{comment}
\usepackage{amssymb}
\usepackage{amsmath}
\usepackage{cite}
\usepackage{url}
\usepackage{xcolor}
\usepackage{cite,graphicx,amsmath,amssymb}
\usepackage{subfigure}
\usepackage{citesort}
\usepackage{fancyhdr}
\usepackage{mdwmath}
\usepackage{mdwtab}
\usepackage{caption}
\usepackage{amsthm}
\usepackage{lipsum}
\usepackage{array}
\usepackage{booktabs}
\usepackage{color}
\usepackage{xcolor}

\newtheorem{theorem}{Theorem}

\newtheorem{lemma}{Lemma}

\newtheorem{corollary}{Corollary}

\newtheorem{remark}{Remark}

\makeatletter
\def\ScaleIfNeeded{%
\ifdim\Gin@nat@width>\linewidth \linewidth \else \Gin@nat@width
\fi } \makeatother

\begin{document}

\title{\Huge{Secrecy Performance Analysis of RIS Assisted Ambient Backscatter Communication Networks}}

\author{Yingjie\ Pei,~\IEEEmembership{Graduate Student Member,~IEEE}, Xinwei\ Yue,~\IEEEmembership{Senior Member,~IEEE}, Chongwen\ Huang,~\IEEEmembership{Senior Member,~IEEE} and Zhiping\ Lu

\thanks{A portion of this work has been submitted to the IEEE Wireless Communications and Networking Conference, Dubai, United Arab Emirates, April, 2024.}
\thanks{Y. Pei is with the Key Laboratory of Information and Communication Systems, Ministry of Information Industry and the Key Laboratory of Modern Measurement $\&$ Control Technology, Ministry of Education, Beijing Information Science and Technology University, Beijing 100101, China. He is also with the State Key Laboratory of Integrated Service Networks, Xidian University, Xi'an 710071, China. (email: yingjie.pei0716@gmail.com).}
\thanks{X. Yue is with the Key Laboratory of Information and Communication Systems, Ministry of Information Industry and also with the Key Laboratory of Modern Measurement $\&$ Control Technology, Ministry of Education, Beijing Information Science and Technology University, Beijing 100101, China (email: xinwei.yue@bistu.edu.cn).}
\thanks{C. Huang is with College of Information Science and Electronic Engineering, Zhejiang University, Hangzhou 310027, China, with the State Key Laboratory of Integrated Service Networks, Xidian University, Xi'an 710071, China, and Zhejiang Provincial Key Laboratory of Info. Proc., Commun. \& Netw. (IPCAN), Hangzhou 310027, China. (email: chongwenhuang@zju.edu.cn).}
\thanks{Z. Lu is the School of Information and Communication Engineering, Beijing University of Posts and Telecommunications, Beijing 100876, China  and also with State Key Laboratory of Wireless Mobile Communications, China Academy of Telecommunications Technology (CATT), Beijing 100191, China.  (email: luzp@cict.com).}
}

\maketitle

\begin{abstract}
Reconfigurable intelligent surface (RIS) and ambient backscatter communication (AmBC) have been envisioned as two promising technologies due to their high transmission reliability as well as energy-efficiency. This paper investigates the secrecy performance of RIS assisted AmBC networks. New closed-form and asymptotic expressions of secrecy outage probability for RIS-AmBC networks are derived by taking into account both imperfect successive interference cancellation (ipSIC) and perfect SIC (pSIC) cases. On top of these, the secrecy diversity order of legitimate user is obtained in high signal-to-noise ratio region, which equals \emph{zero} and is proportional to the number of RIS elements for ipSIC and pSIC, respectively. The secrecy throughput and energy efficiency are further surveyed to evaluate the secure effectiveness of RIS-AmBC networks. Numerical results are provided to verify the accuracy of theoretical analyses and manifest that: i) The secrecy outage behavior of RIS-AmBC networks exceeds that of conventional AmBC networks; ii) Due to the mutual interference between direct and backscattering links, the number of RIS elements has an optimal value to minimise the secrecy system outage probability; and iii) Secrecy throughput and energy efficiency are strongly influenced by the reflecting coefficient and eavesdropper's wiretapping ability.
\end{abstract}
\section{Introduction}
Reconfigurable intelligent surface (RIS) has been deemed a revolutionary technology due to its inherent ability to smartly adjust and reconfigure the wireless channel conditions \cite{2020QingqingMag,2021RIS6G,2021YuanweiRISSurTur}. With the assistance of plentiful low-power reconfigurable elements, the phase shifts and amplitude of incident waves can be regulated independently to establish reliable reflecting links thus realizing the intelligent reconfiguration of signal propagation environment \cite{pei2022secrecy,2021YuanweiRISNOMAOMA}. In \cite{HuangRIS2019}, the authors highlighted the superiority of RIS in terms of energy efficiency by redesigning the transmitting power allocation and the phase shifts.
The authors of \cite{2020BeixiongRISOFDM} verified that RIS aided wireless networks can achieve a higher achievable rate compared to conventional relaying schemes. In \cite{2020TaoqinRIS}, the authors studied that the superior ergodic capacity and outage behaviours can be realized by applying RIS. Despite RIS is capable of providing a more reliable transmission mechanism for wireless communication networks, the radio signals are in danger of being individually overheard by malicious eavesdroppers (Eve) due to the open characteristics and complexity of electromagnetic environment \cite{20206G,20226GMag}. 
With regard to the physical layer security, the secrecy outage probability (SOP) of legitimate users (LUs) was discussed in \cite{2020IRSNOMAPLS}, where RIS is beneficial to secure wireless transmissions. The authors of \cite{2020QingqingRISPLSAN} researched that the secrecy performance of RIS assisted communication networks can be improved by introducing artificial noise especially in multiple Eves cases. In \cite{yan2021intelligent}, the authors enhanced the average secrecy rate within wiretapping channels by leveraging RIS to generate additional randomness. Given the Eve always remains silent, the authors of \cite{2020MIMOIRSsec_2} proposed an RIS aided secure transmission scheme to achieve enhanced secrecy rate with lack of channel state information (CSI). The authors of \cite{2019CuimiaoRISsec} considered more challenging scenarios where the Eve and the LU are in the same direction for the base station (BS). Moreover, the authors of \cite{2022CaihongPLSRISNOMA} further studied the impact of RIS on the secrecy performance of wireless networks in both external and internal wiretapping scenarios.

As a typical green communication technology, ambient backscatter communication (AmBC) has received a lot of attention in recent years since it can modulate its own signal to the surrounding electromagnetic waves for transmissions \cite{2016ABCDetecPerform,darsena2017modeling}. A hybrid design of transmitter based on wireless powered communications and AmBC was proposed in \cite{2018ABCHanzhu}. A cooperative AmBC scheme was proposed in \cite{2018YingchangBC}, where the reader restores signal from both backscatter device as well as the radio frequency (RF). The authors of \cite{2020OPBC} investigated the outage behaviours of AmBC networks with the consideration of co-channel interference. To improve the spectral efficiency, the authors of \cite{2021JiakouRISBC} incorporated novel multi-access strategies into AmBC networks which can realize superior system sum rate. Apart from these, the backscatter signal can also expose to the risk of eavesdropping due to the open nature of wireless propagation \cite{SaadWalid2014PLSBacCom}. In \cite{2021XingwangBCPLS}, artificial noise was utilized to guarantee the secure transmission of AmBC networks by taking hardware impairments into consideration. The authors of \cite{2020XingwangBCNOMAIQ} highlighted the impact of in-phase and quadrature-phase imbalance on the secrecy outage performance of AmBC networks. In \cite{2023ShaoboAmBCSmartTransport}, the secrecy performance of AmBC-aided smart transportation networks was analyzed with the invasion of vicious Eves and jamming scheme was also utilized to confuse the reception at Eves. Furthermore, the authors of \cite{JiYoonHan2020MultiTag} scheduled a portion of tags in a multi-tag AmBC secure communication network for data transmission, while allocating another portion of tags for sending artificial noise or identical data to enhance the security of signals. The authors of \cite{Muratkar2022PLSAmBC} innovatively considered the impact of terminal mobility on the security performance of the AmBC network. In cases where there is no direct connection between the BS and the users, the authors in \cite{XingwangLi2023} investigated the secure transmission characteristics of the AmBC network collaborating with decode-and-forward relays.

The good adaptability of RIS allows it to be effectively deployed in AmBC networks. Specifically, the basic technology principle of RIS-AmBC was introduced in \cite{2022YingchangBacRISTur}. The authors of \cite{2022TurRISBC6G} analysed various cooperative paradigms between RIS and AmBC. The bit error rate performance of RIS-AmBC networks was analyzed in \cite{2021YunfeiRISBC}, where the channel conditions of both source-reader and source-tag links can be improved with the assistance of RIS. The clodsed-form and asymptotic achievable rate expressions of RIS-AmBC networks were obtained in \cite{2023KhanYasinVTC}. As a further advancement, the achievable rate of RIS-AmBC networks was maximized in \cite{2023LokuRISAmBC} by jointly optimizing the phase shifts as RIS. Different from the aforementioned literature where RIS is only used for link enhancement, the authors of \cite{2019BasarRIS} investigated the function of RIS as an access point by utilizing an unmodulated carrier, which indicated that RIS can also perform backscatter communication. On this basis, the authors of \cite{Hancheng2023RISAmBC} divided the RIS into two parts, one for transmitting the primary data signal and the other for transmitting the backscatter signal. Furthermore, authors of \cite{HuiMa2022RISAmBC} exploited a portion of the reflecting elements to perform energy harvesting while the remanent elements implement
backscatter transmission. The effects brought by the number of reflecting elements on the quality of backscatter and direct links were discussed in \cite{2020WenjingBacandDire}, where the backscatter link can be stronger than another link at the expense of deploying sufficient RIS elements and high cost. A secrecy transmission scheme of RIS-AmBC networks was proposed in \cite{JinmingWang2022RISAmBCPLS}, where an alternative algorithm was developed to maximize the secrecy rate of the backscatter signal. Recently, the authors of \cite{BinLyu2023RISAmBCPLS} investigated the secrecy performance of active RIS aided symbiotic radio with imperfect CSI of Eves.

While the interaction between RIS and AmBC networks holds the potential for substantial performance enhancements in low-power communications, the secrecy attributes of both data and backscatter signals within RIS-AmBC networks remain uncertain. Research on RIS-AmBC networks currently only stays at the optimization level. Specifically, the secrecy rate was maximized in \cite{JinmingWang2022RISAmBCPLS}, where the impact brought by the number of users/Eves/RIS elements were discussed at length. As a further advancement, by minimizing the total power consumption while ensuring reliable communication quality for users, a robust and secure active RIS aided symbiotic radio network was proposed in \cite{BinLyu2023RISAmBCPLS}. The two aforementioned papers did not provide an accurate theoretical measurement of the secrecy performance of RIS-AmBC networks. Additionally, there is still a research gap in studying the impacts of different signals, Eve's wiretapping capability, and ipSIC on the secure transmission of RIS-AmBC networks. Hence, we propose a secure RIS-AmBC framework, where the vicious Eve attempts to overhear the data signal and backscatter signal from the ambient RF source and RIS performing as an AmBC device, respectively. Given the low power of the backscatter signal, tuning the reconfigurable elements within the RIS becomes pivotal for coherent alignment of the cascaded channels, thereby amplifying the channel gain of the backscatter links. Moreover, the error propagation and quantization errors during the successive interference cancellation (SIC) process could potentially compromise the decoding efficiency at LU. As a result, thorough research is conducted into the adverse impact of residual interference on the secrecy outage patterns, secrecy throughput, and energy efficiency of RIS-AmBC networks. To the best of our knowledge, this paper is the first to conduct an analysis of the physical layer security performance of RIS-AmBC networks. The main contributions of this paper can be summarized as follows:

\begin{enumerate}
  \item We propose the RIS-AmBC secure communication networks with the existence of a malicious Eve attempting to overhear the information of LU. RIS support facilitates coherent alignment between the incident and backscatter channels of LU, ensuring dependable and secure transmission of backscatter signals. In wiretapping scenarios, the reconfigurable matrix adopts a conventional random phase-shifting design termed an on-off control scheme. Moreover, the SOP serves as a crucial metric in assessing secrecy performance. Furthermore, we survey the superiority of RIS-AmBC networks in the secrecy throughput and secrecy energy efficiency under delay-limited transmission mode.
  \item We derive the closed-form and asymptotic expressions of SOP for LU by taking into account imperfect SIC (ipSIC) and perfect SIC (pSIC). We verified that the SOPs of backscatter signals and the overall RIS-AmBC networks are significantly lower than those of the conventional AmBC networks without RIS. To reap more insights, the secrecy diversity order of LU is further attained, which is equal to \emph{zero} and proportional to the number of RIS elements for ipSIC and pSIC scenarios, respectively.
  \item We investigate the effect of RIS deployment location on the secrecy performance of the RIS-AmBC networks. In general, deploying the RIS in an intermediate position is beneficial for improving the secure transmission of data signal and the overall security of AmBC-RIS networks, but can exacerbate the secrecy outage behaviour of the backscatter signal. Therefore, it is recommended that the RIS should be placed closer to the BS or LU rather than right in the centre of them. Moreover, the growth of both Eve's wiretapping capability as well as reflecting coefficient can be detrimental to the secrecy throughput and secrecy energy efficiency of RIS-AmBC networks.
\end{enumerate}

\subsection{Organization and Notations}
The remainder of this paper is organized as follows. The network model and on-off control scheme are presented in Section \ref{SectionII}. The derivations of SOP, secrecy throughput and energy efficiency are given in Section \ref{SectionIII}. Numerical results and analyses are provided in Section \ref{SectionIV} followed by conclusions provided in Section \ref{SectionV}.

The primary notations adopted in this paper are explained as follows. $\mathbb{E}\left\{  \cdot  \right\}$ indicates the expectation operation. The cumulative distribution function (CDF) and the probability density function (PDF) with parameter \emph{X} are given by ${F_X}\left(  \cdot  \right)$ and ${f_X}\left(  \cdot  \right)$, respectively. ${{\mathbf{I}}_p }$ refers to a $p  \times p $ identity matrix and ${{\mathbf{1}}_q }$ denotes a $q  \times 1$ all-ones column vector. $ \otimes $ represents Kronecker product. ${{\mathbf{h}}^H}$ indicates the hermitian transpose operation on a matrix ${{\mathbf{h}}}$. $h \sim \mathcal{C}\mathcal{N}\left( {a,b} \right)$ represents variable $h$ follows a complex Gaussian distribution with parameters $a$ and $b$.

\section{Network Model}\label{SectionII}
\subsection{Signal Model}

\begin{figure}[t!]
    \begin{center}
        \includegraphics[width=3in,  height=2in]{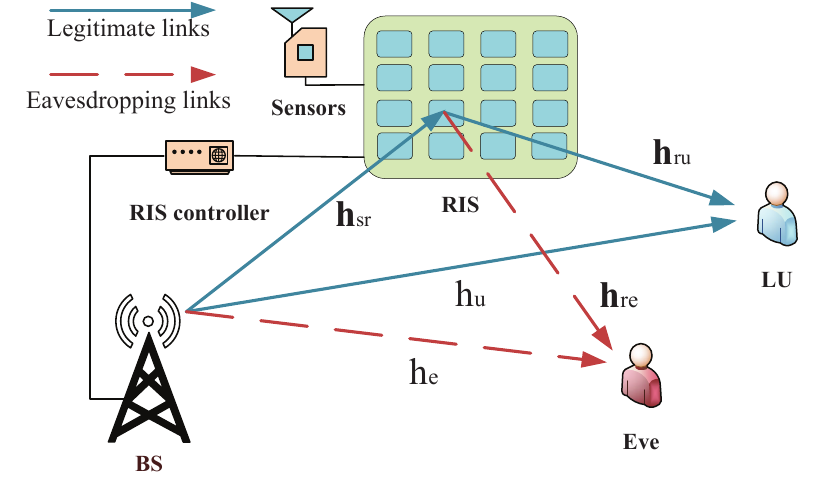}
        \caption{An illustration of RIS assisted secure AmBC networks.}
        \label{sys_model.eps}
    \end{center}
\end{figure}

We consider a secure RIS aided AmBC network, where the ambient RF source, LU and the Eve are all equipped with single antenna\footnote{Since the BS usually has fixed number of antennas when deploying RIS in the AmBC networks, we downplay the effect of it and focus on the influence of RIS and Eve characteristics on the secrecy performance of AmBC networks. The secure transmission of RIS-AmBC with multiple antennas is set aside for our future treatise.}. The RIS is composed of \emph{M} passive RIS elements and each element can independently adjust the phase shift of incident wave to reconfigure wireless propagation environments. In secure RIS-AmBC networks, the RIS acts as an AmBC device which can modulate its binary signal onto incident waves without any dedicated RF source. The complex channel gains from the source to RIS, and then from RIS to LU and Eve are indicated as ${{\mathbf{h}}_{sr}} \in {\mathbb{C}^{M \times 1}}$, ${{\mathbf{h}}_{ru}} \in {\mathbb{C}^{M \times 1}}$ and ${{\mathbf{h}}_{re}} \in {\mathbb{C}^{M \times 1}}$, respectively. ${h_u} \in {\mathbb{C}^{1 \times 1}}$ and ${h_e} \in {\mathbb{C}^{1 \times 1}}$ represent the direct communication links from the source to LU and Eve, respectively. Assuming that reflecting elements on the RIS are small in size and densely arranged, ${{\mathbf{h}}_{sr}}$, ${{\mathbf{h}}_{ru}}$ and ${{\mathbf{h}}_{re}}$ can be regarded as identically distributed channels. Note that the wireless links involved in the RIS-AmBC networks are supposed to be Rayleigh flat fading channels
\cite{2021XingwangBCPLS,2020WenjingBacandDire}. We suppose that the CSI of LU can be acquired perfectly at RIS \cite{2021ChongwenIRSChannEsti}, but not the Eve due to its silence.

As a consequence, the LU is able to receive two parts of the signal, i.e., the data signal from the direct link of BS and the backscatter signal from RIS. The received information at LU is given by
\begin{align}\label{received signal at user}
{y_u} = \left[  {{h_u} + \kappa {\mathbf{h}}_{ru}^H{\mathbf{\Phi }}{{\mathbf{h}}_{sr}}c(t)} \right]\sqrt {{P_s}}x(t) + {n_u},
\end{align}where ${{P_s}}$ denotes the transmitting power of the source. ${\mathbf{\Phi }} = {\text{diag}}\left\{ {{e^{j{\theta _1}}}, \cdots ,{e^{j{\theta _m}}}, \cdots ,{e^{j{\theta _M}}}} \right\}$ denotes the phase shift matrix at RIS and ${{\theta _m}}$ is the phase shift at the \emph{m}-th RIS elements. $\kappa  \in \left( {0,1} \right)$ is the reflecting coefficient. $c$ and $x$ indicate the backscatter signals from RIS and data signal from the source with normalized power, i,e., $\mathbb{E}\left\{ {{{\left| c(t) \right|}^2}} \right\} = \mathbb{E}\left\{ {{{\left| x(t) \right|}^2}} \right\} = 1$. ${n_u} \sim \mathcal{C}\mathcal{N}\left( {0,\sigma _u^2} \right)$ is the additive white Gaussian noise (AWGN) at LU. 
We assume that the coherent phase shift is utilized to redesign phase shifts at RIS to match with phases of cascaded channels\footnote{Coherent phase shift is considered in this paper since the theoretical results can be provided as upper bound. Moreover, the effect of continuous phase shift can be approximated by increasing the discrete phase bit and optimization method, which is beyond the scope of this paper.}. Based on these, the above equation can be further given by
\begin{align}\label{received signal at user v2 coherent}
{y_u} = \left[ {{h_u} + \sum\limits_{m = 1}^M {\left| {h_{ru}^mh_{sr}^m} \right|} c(t)} \right]\sqrt {{P_s}}x(t) + {n_u},
\end{align}where $h_{u} = \sqrt {\eta d_{su}^{ - \lambda }} \tilde h_{u}$, $h_{sr}^m = \sqrt {\eta d_{sr}^{ - \lambda }} \tilde h_{sr}^m$, $h_{ru}^m = \sqrt {\eta d_{ru}^{ - \lambda }} \tilde h_{ru}^m$ denote the channel coefficients from the source to the \emph{m}-th RIS element and from the \emph{m}-th RIS element to LU, respectively. $h_{u} \sim \mathcal{C}\mathcal{N}\left( {0,{\Omega _{u}}} \right)$, $h_{sr}^m \sim \mathcal{C}\mathcal{N}\left( {0,{\Omega _{sr}}} \right)$, $h_{ru}^m \sim \mathcal{C}\mathcal{N}\left( {0,{\Omega _{ru}}} \right)$ and $\tilde h_{u}$, $\tilde h_{sr}^m$, $\tilde h_{ru}^m \sim \mathcal{C}\mathcal{N}\left( {0,1} \right)$. The distances of source-RIS, RIS-LU and source-LU are denoted as ${d_{sr}}$, ${d_{ru}}$ and ${d_{su}}$, respectively. $\lambda $ indicates the path loss exponent and $\eta $ is the frequency dependent factor.

SIC is deemed as a effective technology for interference cancellation and extensively researched and applied in various wireless communication networks [43,49,50]. Considering SIC primarily relies on differences in signal strength to enable a more accurate decoding of received signals, in the proposed AmBC-RIS networks, the LU decodes the data signal from direct links with a higher signal-to-interference-plus-noise ratio(SINR) first and then employs SIC to remove it\footnote{Considering the weakness of the backscatter signal and the limited RIS reflector, we assume that the quality of the direct link is stronger than that of the backscatter link. In addition, multiplicative fading effects further weaken the backscattered signal \cite{LinglongDai2023RIS}. For the case where the direct link is stronger than the backscatter link, please refer to \cite{2020WenjingBacandDire}.}. The SINR for LU to decode data signal can be given by
\begin{align}\label{SINR uu}
{\gamma _{uu}} = \frac{{{P_s}{{\left| {{h_u}} \right|}^2}}}{{{\kappa ^2}{{\left( {\sum\nolimits_{m = 1}^M {\left| {h_{ru}^mh_{sr}^m} \right|} } \right)}^2}{P_s} + \sigma _u^2}}.
\end{align} After this, the LU starts decoding the backscatter signal and the relevant SINR is shown as
\begin{align}\label{SINR uc ipsic}
\gamma _{uc}^{ipSIC} = \frac{{{\kappa ^2}{{\left( {\sum\nolimits_{m = 1}^M {\left| {h_{ru}^mh_{sr}^m} \right|} } \right)}^2}{P_s}}}{{\varpi {P_s}{{\left| {{h_{ipu}}} \right|}^2} + \sigma _u^2}},
\end{align}where $\varpi  \in \left[ {0,1} \right]$ denotes the residual interference degree of SIC. $\varpi  = 0$ and $\varpi  \ne 0$ represent the operation of pSIC and ipSIC, respectively. Without loss of generality, the corresponding factor of residual interference brought by ipSIC is considered as ${h _{ipu}} \sim \mathcal{C}\mathcal{N}\left( {0,{\Omega _{ipu}}} \right)$ \cite{2022XinweiYueRISNOMA}.

Due to the open nature of wireless communications, confidential signals are exposed to the risk of being overheard by the baleful Eve. It is pivotal to evaluate the secrecy characteristics of RIS-AmBC networks with the presence of a wiretap. In this case, the received signal at Eve is shown as
\begin{align}\label{received signal at Eve}
{y_e} = \left( {{h_e} + {\mathbf{h}}_{re}^H{\mathbf{\Phi }}{{\mathbf{h}}_{sr}}c} \right)\sqrt {{P_s}}x + {n_e},
\end{align}where ${n_e} \sim \mathcal{C}\mathcal{N}\left( {0,\sigma _e^2} \right)$ is the AWGN at the Eve. $h_{e} = \sqrt {\eta d_{se}^{ - \lambda }} \tilde h_{e}$, ${{\mathbf{h}}_{re}} = {\left[ {h_{re}^1, \cdots ,h_{re}^m, \cdots ,h_{re}^M} \right]^H}$, $h_{re}^m = \sqrt {\eta d_{re}^{ - \lambda }} \tilde h_{re}^m$ and $\tilde h_{e}$, $\tilde h_{re}^m \sim \mathcal{C}\mathcal{N}\left( {0,1} \right)$. The distances of RIS-Eve and source-Eve are denoted as ${d_{re}}$ and ${d_{se}}$, respectively.

Similar to the SIC processes of LU, the Eve first decodes the data signal and then obtains the backscatter signal. The corresponding SINRs can be separately given by
\begin{align}\label{SINR Eve decode u}
{\gamma _{eu}} = \frac{{{P_s}{{\left| {{h_e}} \right|}^2}}}{{{\kappa ^2}{{\left| {{\mathbf{h}}_{re}^H{\mathbf{\Phi }}{{\mathbf{h}}_{sr}}} \right|}^2}{P_s} + \sigma _e^2}},
\end{align}
and
\begin{align}\label{SINR Eve decode c}
\gamma _{ec}^{ipSIC} = \frac{{{\kappa ^2}{{\left| {{\mathbf{h}}_{re}^H{\mathbf{\Phi }}{{\mathbf{h}}_{sr}}} \right|}^2}{P_s}}}{{\varpi {P_s}{{\left| {{h_{ipe}}} \right|}^2} + \sigma _e^2}},
\end{align}where the residential interference at Eve is denoted by ${h_{ipe}} \sim \mathcal{C}\mathcal{N}\left( {0,{\Omega _{ipe}}} \right)$.

\subsection{RIS-AmBC Networks with On-off Control Scheme}
Different from that RIS can adjust the phase shifts to align the phases of the cascaded channels for LU, the coherent phase shift is inapplicable to the Eve due to the lack of CSI. To tackle this problem, a random phase shift called on-off control scheme is harnessed to design phase shifts of RIS to the Eve where each element of RIS can be set to 1 (on) or 0 (off) \cite{2019ZhiguoAsimpledesign}. Specifically, we suppose that the total number of reflecting components $M = P \times Q$, where $P$ and $Q$ are both positive integers\footnote{By altering the relative values of \emph{P} and \emph{Q}, the on-off control scheme can approximate different wiretapping scenarios. For instance, large value of \emph{P} indicates Eve may be positioned farther from the RIS beam, making it nearly impossible to acquire sufficient leaked signals. Please refer \cite{2019ZhiguoAsimpledesign} for more information about on-off control scheme. }. Setting the matrix ${\mathbf{V}} = {{\mathbf{I}}_P} \otimes {{\mathbf{1}}_Q}$ and the \emph{p}-th column of ${\mathbf{V}}$ is denoted as a vector ${{\mathbf{v}}_p}$.
The corresponding cascaded channel ${{\mathbf{h}}_{re}^H{\mathbf{\Phi }}{{\mathbf{h}}_{sr}}}$ can be recast as ${{\mathbf{v}}_p^H{{\mathbf{\Lambda }}_{re}}{{\mathbf{h}}_{sr}}}$, where the diagonal matrix ${{{\mathbf{\Lambda }}_{re}}}$ is generated by diagonalizing ${{\mathbf{h}}_{re}^H}$. With the assistance of on-off control scheme, SINRs for the Eve to intercept data and backscatter signals can be rewritten as ${\gamma _{eu}} = \frac{{{P_s}{{\left| {{h_e}} \right|}^2}}}{{{\kappa ^2}{{\left| {{\mathbf{v}}_p^H{{\mathbf{\Lambda }}_{re}}{{\mathbf{h}}_{sr}}} \right|}^2}{P_s} + \sigma _e^2}}$ and $\gamma _{ec}^{ipSIC} = \frac{{{\kappa ^2}{{\left| {{\mathbf{v}}_p^H{{\mathbf{\Lambda }}_{re}}{{\mathbf{h}}_{sr}}} \right|}^2}{P_s}}}{{\varpi {P_s}{{\left| {{h_{ipe}}} \right|}^2} + \sigma _e^2}}$, respectively.

\begin{lemma} \label{Lemma1}
By utilizing the coherent phase shift, the cumulative distribution function (CDF) of SINR for LU to decode the data signal in RIS-AmBC networks is given by
\begin{align}\label{CDF SINR uu}
{F_{{\gamma _{uu}}}}\left( x \right)  \approx  1 - \frac{1}{{\Gamma \left( {\alpha  + 1} \right)}}\sum\limits_{d = 0}^D {{G_d}{{\left( {{\tau _d}} \right)}^\alpha }{e^{ - \frac{{x\left( {{\kappa ^2}{\beta ^2}\tau _d^2\rho  + 1} \right)}}{{{\Omega _u}\rho }}}}},
\end{align}
where $\rho  = {{{P_s}} \mathord{\left/
 {\vphantom {{{P_s}} {\sigma _u^2}}} \right.
 \kern-\nulldelimiterspace} {\sigma _u^2}}$ refers to the transmitting signal-to-noise ratio (SNR). $\alpha  = \left[ {{\pi ^2}M/\left( {16 - {\pi ^2}} \right)} \right] - 1$ and $\beta  = \left( {4/\pi  - \pi /4} \right)\sqrt {{\Omega _{sr}}{\Omega _{ru}}} $. ${G_d} = {{{\left( {D!} \right)}^2}}/\left\{{{{{\tau _d}}{{\left[ {L_{_D}^\prime \left( {{\tau _d}} \right)} \right]}^2}}}\right\}$ 
 and ${{\tau _d}}$ represent the weight of Gauss-Laguerre quadrature formula and the d-th zero point of Laguerre polynomial ${L_D}\left( {{\tau _d}} \right)$ with d = 1,2,3, $\cdots$, D, respectively. D presents a complexity accuracy tradeoff parameter and the equal sign in (\ref{CDF SINR uu}) can be established when D approaches infinity \emph{\cite{2022XinweiYueRISNOMA,Yingjie2023RISNOMAPLS}}. Note that D is set to 300 for the purpose of sufficient convergence. $\Gamma \left(  \cdot  \right)$ indicates the gamma function \emph{\cite[Eq. (8.310.1)]{gradvstejn2000table}}.
 \end{lemma}
\begin{proof}
\emph{The CDF of SINR for LU to decode the data signal in RIS-AmBC networks can be written as}

\begin{align}\label{a1}
{F_{{\gamma _{uu}}}}\left( x \right) &= {\rm{Pr}}\left( {{\gamma _{uu}} < x} \right)\notag \\ &= {\rm{Pr}}[\underbrace {{{\left| {{h_u}} \right|}^2}}_Z < x({\kappa ^2}|\underbrace {\sum\nolimits_{m = 1}^M {\left| {h_{ru}^mh_{sr}^m} \right|} }_Y{|^2}\rho + 1)/\rho ]\notag \\ &= \int_0^\infty  {{f_Y}\left( y \right){F_Z}\left[ {x\left( {{\kappa ^2}{y^2}\rho + 1} \right)/\rho } \right]} dy,
\end{align} \emph{where} ${F_{Z}}\left( z \right) = 1 - {e^{ - \frac{z}{{{\Omega _u}}}}}$ \emph{based on} \cite{liu2014outage}.

\emph{According to} \cite{2020Boulogeorgos,primak2005stochastic}, \emph{the probability density function (PDF) of $Y = {\sum\nolimits_{m = 1}^M {\left| {h_{ru}^mh_{sr}^m} \right|} }$ can be given as}
\begin{align}\label{a2}
{f_Y}\left( y \right) = \frac{{{y^\alpha }}}{{{\beta ^{\alpha  + 1}}\Gamma \left( {\alpha  + 1} \right)}}{e^{ - \frac{y}{\beta }}},
\end{align}\emph{where $\alpha  = \left[ {{\pi ^2}M/\left( {16 - {\pi ^2}} \right)} \right] - 1$ and $\beta  = \left( {4/\pi  - \pi /4} \right)\sqrt {{\Omega _{sr}}{\Omega _{ru}}} $. Upon substituting (\ref{a2}) into (\ref{a1}), we can acquire}
\begin{align}\label{a3}
{F_{{\gamma _{uu}}}}\left( x \right) = 1 - \int_0^\infty  {\frac{{{y^\alpha }{e^{ - \frac{y}{\beta }}}}}{{{\beta ^{\alpha  + 1}}\Gamma \left( {\alpha  + 1} \right)}}{e^{ - \frac{{x\left( {{\kappa ^2}{y^2}\rho  + 1} \right)}}{{{\Omega _u}\rho }}}}} dy.
\end{align}
\emph{Assuming $t = y/\beta $ and utilizing Gauss-Laguerre integration, we can obtain (\ref{CDF SINR uu}) and the proof is completed.}
\end{proof}
\begin{lemma} \label{Lemma2}
By utilizing the coherent phase shift, the CDF of SINR for LU to decode the backscatter signal with ipSIC in RIS-AmBC networks is given by
\begin{align}\label{CDF SINR uc ipsic}
{F_{\gamma _{uc}^{ipSIC}}}\left( x \right) \approx \sum\limits_{d = 0}^D {\frac{{{G_d}}}{{\Gamma \left( {\alpha  + 1} \right)}}\gamma \left( {\alpha  + 1,\sqrt {\frac{{x\left( {\varpi \rho {\Omega _{ipu}}{\tau _d} + 1} \right)}}{{\rho {\kappa ^2}{\beta ^2}}}} } \right)} ,
\end{align}
where $\gamma \left(  \cdot  \right)$ is the lower incomplete gamma function \emph{\cite[Eq. (8.350.1)]{gradvstejn2000table}}.
 \end{lemma}

\begin{proof}
\emph{According to the analyses above, the SINR for LU to decode backscatter signal with ipSIC can be expressed as
\begin{align}\label{b1}
\gamma _{uc}^{ipSIC} = \frac{{{\kappa ^2}{{\left( {\sum\nolimits_{m = 1}^M {\left| {h_{ru}^mh_{sr}^m} \right|} } \right)}^2}\rho }}{{\varpi \rho {{\left| {{h_{ipu}}} \right|}^2} + 1}}.
\end{align}
Hence, the CDF of ${\gamma _{uc}^{ipSIC}}$ can be given as
\begin{align}\label{b2}
{F_{\gamma _{uc}^{ipSIC}}}\left( x \right) &= \Pr \left( {\gamma _{uc}^{ipSIC} < x} \right)\notag \\ &= \Pr [ {{|\underbrace {\sum\limits_{m = 1}^M {|h_{ru}^mh_{sr}^m|} }_Y{|^2}} < {\frac{x}{{\rho {\kappa ^2}}}}( {\varpi \rho \underbrace {|{h_{ipu}}{|^2}}_Z + 1} )} ] \notag \\ &= \int_0^\infty  {{f_Z}\left( z \right)} {F_Y}\left[ {\sqrt {x\left( {\varpi \rho z + 1} \right)/\left( {\rho {\kappa ^2}} \right)} } \right]dz,
\end{align}where ${f_Z}\left( z \right) = \frac{1}{{{\Omega _{ipu}}}}{e^{ - \frac{z}{{{\Omega _{ipu}}}}}}$ \emph{\cite{yue2017exploiting}}.
Applying the integration operation to (\ref{a2}), the CDF of $Y = {\sum\limits_{m = 1}^M {|h_{ru}^mh_{sr}^m|} }$ is shown as
\begin{align}\label{b3}
{F_Y}\left( y \right) = \frac{{\gamma \left( {\alpha  + 1,\frac{y}{\beta }} \right)}}{{\Gamma \left( {\alpha  + 1} \right)}}.
\end{align}
Upon substituting ${f_Z}\left( z \right)$ and (\ref{b3}) into (\ref{b2}), ${F_{\gamma _{uc}^{ipSIC}}}$ can be further written as
\begin{align}\label{b4}
{F_{\gamma _{uc}^{ipSIC}}}\left( x \right) =& \frac{1}{{\Gamma \left( {\alpha  + 1} \right){\Omega _{ipu}}}}\int_0^\infty  {{e^{ - \frac{z}{{{\Omega _{ipu}}}}}}}\notag \\ &\times \gamma \left( {\alpha  + 1,\frac{1}{{\kappa \beta }}\sqrt {x\left( {\varpi \rho z + 1} \right)/\rho } } \right)dz.
\end{align}By replacing $z/{{\Omega _{ipu}}}$ with $t$ and adopting Gauss-Laguerre integration, (\ref{CDF SINR uc ipsic}) is acquired which completes the proof.}
\end{proof}
When $\varpi  = 0$, the CDF of SINR for LU to decode the backscatter signal with pSIC in RIS-AmBC networks is given by
\begin{align}\label{CDF SINR uc psic}
{F_{\gamma _{uc}^{pSIC}}}\left( x \right) = \frac{1}{{\Gamma \left( {\alpha  + 1} \right)}}\gamma \left( {\alpha  + 1,\frac{1}{{\beta \kappa }}\sqrt {\frac{x}{\rho }} } \right).
\end{align}

\section{Secrecy Performance Evaluation}\label{SectionIII}
In this section, the SOP is selected as a critical indicator to evaluate the secrecy performance of RIS-AmBC networks. The secrecy diversity order is also derived in the high SNR region to highlight the approximate secrecy features.
\subsection{Statistical Properties for Eavesdropping Channels}
For the purpose of evaluating the negative effects of Eve imposed on the secrecy RIS-AmBC transmission, the statistical properties for eavesdropping channels are analysed in the following.

By applying on-off control scheme, (\ref{SINR Eve decode u}) and (\ref{SINR Eve decode c}) can be rewritten as
\begin{align}\label{SINR eu 1bit}
{{\hat \gamma }_{eu}} = \frac{{{P_s}{{\left| {{h_e}} \right|}^2}}}{{{\kappa ^2}{{\left| {{\mathbf{v}}_p^H{{\mathbf{\Lambda }}_{re}}{{\mathbf{h}}_{sr}}} \right|}^2}{P_s} + \sigma _e^2}},
\end{align}
and
\begin{align}\label{SINR ec 1bit}
\hat \gamma _{ec}^{ipSIC} = \frac{{{\kappa ^2}{{\left| {{\mathbf{v}}_p^H{{\mathbf{\Lambda }}_{re}}{{\mathbf{h}}_{sr}}} \right|}^2}{P_s}}}{{\varpi {P_s}{{\left| {{h_{ipe}}} \right|}^2} + \sigma _e^2}},
\end{align}respectively.

\begin{lemma} \label{Lemma3}
By utilizing the on-off control scheme, the CDF of SINR for the Eve to decode data signal in RIS-AmBC networks is given by
\begin{align}\label{CDF SINR eu}
{F_{{{\hat \gamma }_{eu}}}}\left( x \right) =& 1 - \frac{{2{\Omega _e}}}{{\Gamma \left( Q \right){\kappa ^2}x}}\sum\limits_{d = 0}^D {{G_d}}\notag \\ &\times \sqrt {\frac{{\Xi _{eu}^{Q - 1}}}{{{{\left( {{\Omega _{sr}}{\Omega _{re}}} \right)}^{Q + 1}}}}} {K_{Q - 1}}\left( {2\sqrt {\frac{{{\Xi _{eu}}}}{{{\Omega _{sr}}{\Omega _{re}}}}} } \right),
\end{align}
where ${\Xi _{eu}} = \left( {{\Omega _e}{\rho _e}{\tau _d} - x} \right)/\left( {{\kappa ^2}{\rho _e}x} \right)$, ${\rho _e} = {P_s}/\sigma _e^2$ denotes the transmitting SNR of eavesdropping channels.
\end{lemma}

\begin{proof} \label{Proof of lemma3}
\emph{The CDF of SINR for the Eve to decode data signal can be formulate as follows
\begin{align}\label{Proof of lemma3 1}
{F_{{{\hat \gamma }_{eu}}}}\left( x \right) &= \Pr \left( {{{\hat \gamma }_{eu}} < x} \right)\notag \\ &= \Pr \left[ {{{\left| {{h_e}} \right|}^2} < \frac{{x\left( {{\kappa ^2}{{\left| {{\mathbf{v}}_p^H{{\mathbf{\Lambda }}_{re}}{{\mathbf{h}}_{sr}}} \right|}^2}{P_s} + \sigma _e^2} \right)}}{{{P_s}}}} \right].
\end{align}
We suppose that $Z = {\left| {{\mathbf{v}}_p^H{{\mathbf{\Lambda }}_{re}}{{\mathbf{h}}_{sr}}} \right|^2}$, (\ref{Proof of lemma3 1}) can be written as
\begin{align}\label{Proof of lemma3 2}
{F_{{{\hat \gamma }_{eu}}}}\left( x \right) =& \int_0^\infty  {{F_{{{\left| {{h_e}} \right|}^2}}}\left[ {x\left( {{\kappa ^2}{P_s}z + \sigma _e^2} \right)/{P_s}} \right]} {f_Z}\left( z \right)dz \notag \\ =& \int_0^\infty  {\left[ {1 - {e^{ - x\left( {{\kappa ^2}{\rho _e}z + 1} \right)/\left( {{\Omega _e}{\rho _e}} \right)}}} \right]} {f_Z}\left( z \right)dz.
\end{align}
According to \emph{\cite{liu2014outage}}, the PDF of cascaded channels $Z = {\left| {{\mathbf{v}}_p^H{{\mathbf{\Lambda }}_{re}}{{\mathbf{h}}_{sr}}} \right|^2}$ is shown as
\begin{align}\label{Proof of lemma3 3}
{f_Z}\left( z \right) = \frac{2}{{\Gamma \left( Q \right)}}\sqrt {\frac{{{z^{Q - 1}}}}{{{{\left( {{\Omega _{sr}}{\Omega _{re}}} \right)}^{Q + 1}}}}} {K_{Q - 1}}\left( {2\sqrt {\frac{z}{{{\Omega _{sr}}{\Omega _{re}}}}} } \right).
\end{align}
Upon substituting (\ref{Proof of lemma3 3}) into (\ref{Proof of lemma3 2}), we can obtain
\begin{align}\label{Proof of lemma3 4}
{F_{{{\hat \gamma }_{eu}}}}\left( x \right) =& 1 - \frac{2}{{\Gamma \left( Q \right)}}\int_0^\infty  {{e^{ - \frac{{x\left( {{\kappa ^2}{\rho _e}z + 1} \right)}}{{{\Omega _e}{\rho _e}}}}}}\notag\\  &\times \sqrt {\frac{{{z^{Q - 1}}}}{{{{\left( {{\Omega _{sr}}{\Omega _{re}}} \right)}^{Q + 1}}}}} {K_{Q - 1}}\left( {2\sqrt {\frac{z}{{{\Omega _{sr}}{\Omega _{re}}}}} } \right)dz.
\end{align}
Setting $t = \left[ {x\left( {{\kappa ^2}{\rho _e}z + 1} \right)} \right]/\left( {{\Omega _e}{\rho _e}} \right)$ and $z = \left( {{\Omega _e}{\rho _e}t - x} \right)/\left( {{\kappa ^2}{\rho _e}x} \right)$. Hence, (\ref{Proof of lemma3 4}) can be transformed into
\begin{align}\label{Proof of lemma3 5}
{F_{{{\hat \gamma }_{eu}}}}\left( x \right) =& 1 - \int_0^\infty  {\frac{{2{\Omega _e}{e^{ - t}}}}{{\Gamma \left( Q \right){\kappa ^2}x}}} {K_{Q - 1}}\left( {2\sqrt {\frac{{{\Omega _e}{\rho _e}t - x}}{{{\kappa ^2}{\rho _e}x{\Omega _{sr}}{\Omega _{re}}}}} } \right) \notag \\ &\times \sqrt {\frac{{{{\left( {{\Omega _e}{\rho _e}t - x} \right)}^{Q - 1}}}}{{{{\left( {{\kappa ^2}{\rho _e}x} \right)}^{Q - 1}}{{\left( {{\Omega _{sr}}{\Omega _{re}}} \right)}^{Q + 1}}}}}dt.
\end{align}
With the assistance of Gauss-Laguerre integration, we can acquire (\ref{CDF SINR eu}) and the proof is completed.}
\end{proof}

\begin{lemma} \label{Lemma4}
By utilizing the on-off control scheme, the CDF of SINR for the Eve to decode backscatter signal with ipSIC in RIS-AmBC networks is given by
\begin{align}\label{CDF SINR ec ipsic}
{F_{\hat \gamma _{ec}^{ipSIC}}}\left( x \right) = 1 - \sum\limits_{d = 0}^D {\frac{{2{G_d}}}{{\Gamma \left( Q \right)}}} {\left( {\Xi _{ec}^{ipSIC}} \right)^{\frac{Q}{2}}}{K_Q}\left( {2\sqrt {\Xi _{ec}^{ipSIC}} } \right),
\end{align}
where $\Xi _{ec}^{ipSIC} = \left[ {x\left( {\varpi {\rho _e}{\Omega _{ipe}}{\tau _d} + 1} \right)} \right]/\left( {{\kappa ^2}{\rho _e}{\Omega _{sr}}{\Omega _{re}}} \right)$.
\end{lemma}

When $\varpi  = 0$, the CDF of SINR for the Eve to intercept backscatter signal with pSIC can be given by
\begin{align}\label{CDF SINR ec psic}
{F_{\hat \gamma _{ec}^{pSIC}}}\left( x \right) = 1 - \frac{2}{{\Gamma \left( Q \right)}}{\left( {\Xi _{ec}^{pSIC}} \right)^{\frac{Q}{2}}}{K_Q}\left( {2\sqrt {\Xi _{ec}^{pSIC}} } \right),
\end{align} where $\Xi _{ec}^{pSIC} = x/\left( {{\kappa ^2}{\rho _e}{\Omega _{sr}}{\Omega _{re}}} \right)$.

\subsection{Secrecy Outage Probability Analysis}
Assuming the secrecy outage event occurs when the secrecy rate for LU to decode data signal and backscatter signal is less than the target secrecy rate, i,e., ${C_\zeta } = {\left[ {\log \left( {1 + {\gamma _{u\zeta }}} \right) - \log \left( {1 + {\gamma _{e\zeta }}} \right)} \right]^ + } < {R_\zeta }$, where $\zeta  \in \left\{ {u,c} \right\}$ and ${\left[ a \right]^ + } = \max \left( {a,0} \right)$. Note that the difference between ${\gamma _{u\zeta }}$ and ${\gamma _{e\zeta }}$ is primarily generated by the variables ${{\mathbf{h}}_{ru}}$ and ${{\mathbf{h}}_{re}}$, which are independent of each other \cite{Yingjie2023RISNOMAPLS,2020IRSNOMAPLS}. In addition, the residential interference caused by ipSIC is related to the hardware characteristics of the receiver which can be also regarded as independent variables, which ensures the correctness of subsequent SOP and throughput derivation processes. Hence, the SOP of users can be further calculated by utilizing marginal distributions in the following steps. The SOP expression for LU to detect signal $\zeta$ is given by
\begin{align}\label{SOP 1}
{P_\zeta } &= {\rm{Pr}}\left[ {{\gamma _{u\zeta }} < {2^{{R_\zeta }}}\left( {1 + {\gamma _{e\zeta }}} \right) - 1} \right]\notag \\ &= \int_0^\infty  {{f_{e\zeta }}\left( x \right){F_{\gamma _{u\zeta }^\psi }}\left[ {{2^{{R_\zeta }}}\left( {1 + x} \right) - 1} \right]} dx,
\end{align}where $\psi  \in \left\{ {\rm{ipSIC,pSIC}} \right\}$, $\zeta  \in \left\{ {u,c} \right\}$. ${F_{\gamma _{u\zeta }^\psi }}\left( x \right)$ can be acquired from Lemma \ref{Lemma1} and Lemma \ref{Lemma2}. ${{f_{e\zeta }}\left( x \right)}$ denotes the PDF of the SINR for the Eve to decode signal $\zeta$ which can be obtained by taking the derivative of (\ref{CDF SINR eu}), (\ref{CDF SINR ec ipsic}) and (\ref{CDF SINR ec psic}). It can be seen that the integral formula in (\ref{SOP 1}) can be difficult to solve precisely. An approximate method is developed to obtained the asymptotic SOP expression with a verified accuracy.
\begin{theorem} \label{Theorem1}
The SOP of LU to decode data signal in RIS-AmBC networks can be approximated as
\begin{align}\label{SOP uu}
{P_u}\left( {{R_u}} \right) \approx 1 - \sum\limits_{d = 0}^D {\frac{{G_d}{{\left( {{\tau _d}} \right)}^\alpha }}{{\Gamma \left( {\alpha  + 1} \right)}}{e^{ - \frac{{{\varepsilon _{u}}\left( {{\kappa ^2}{\beta ^2}\tau _d^2\rho  + 1} \right)}}{{{\Omega _u}\rho }}}}},
\end{align}where ${\varepsilon _{u}} = {2^{{R_u}}}\left[ {1 + \left( {{\rho _e}{\Omega _e}} \right)/\left( {{\kappa ^2}Q{\Omega _{sr}}{\Omega _{re}}{\rho _e} + 1} \right)} \right] - 1$.
\end{theorem}
\begin{proof}
\emph{See Appendix A.}
\end{proof}
\begin{theorem}\label{Theorem2}
The SOP of LU to decode the backscatter signal with ipSIC in RIS-AmBC networks can be approximated as
\begin{align}\label{SOP uc ipsic}
P_c^{ipSIC}\left( {{R_c}} \right) \approx \sum\limits_{d = 0}^D {\frac{{{G_d}}}{{\Gamma \left( {\alpha  + 1} \right)}}\gamma \left[ {\alpha  + 1,\sqrt {\frac{{{\varepsilon _{c1}}\left( {\varpi \rho {\Omega _{ipu}}{\tau _d} + 1} \right)}}{{{\beta ^2}}}} } \right]} ,
\end{align}where ${\varepsilon _{c1}} = \left\{ {{2^{{R_c}}}\left[ {1 + \left( {{\kappa ^2}{\rho _e}Q{\Omega _{sr}}{\Omega _{re}}} \right)/\left( {\varpi {\rho _e}{\Omega _{ipe}} + 1} \right)} \right] - 1} \right\}$ $/\left( {{\kappa ^2}\rho } \right)$.
\end{theorem}
When $\varpi  = 0$, the SOP of LU to decode the backscatter signal with pSIC in RIS-AmBC networks can be approximated as
\begin{align}\label{SOP uc psic}
P_c^{pSIC}\left( {{R_c}} \right) \approx \frac{1}{{\Gamma \left( {\alpha  + 1} \right)}}\gamma \left( {\alpha  + 1,\frac{1}{\beta }\sqrt {{\varepsilon _{c2}}} } \right),
\end{align}where ${\varepsilon _{c2}} = \left[ {{2^{{R_c}}}\left( {1 +  {{\kappa ^2}{\rho _e}Q{\Omega _{sr}}{\Omega _{re}}} } \right) - 1} \right]$ $/\left( {{\kappa ^2}\rho } \right)$.

\subsection{Secrecy Diversity Order}
To reap more insights, the secrecy performance of RIS-AmBC networks is further analyzed in high-SNR region. The secrecy diversity order can be defined as follows \cite{2022XinweiYueRISNOMA}
\begin{align}\label{div define}
Div =  - \mathop {\lim }\limits_{\rho  \to \infty } \frac{{\log \left[ {{P_\infty }\left( \rho  \right)} \right]}}{{\log \left( \rho  \right)}},
\end{align}where ${{P_\infty }\left( \rho  \right)}$ refers to the asymptotic SOP at high SNRs.
\begin{corollary}\label{Corollary1}
Conditioned on $\rho  \to \infty $, the asymptotic SOP of LU to decode the data signal in RIS-AmBC networks can be approximated by
\begin{align}\label{asy sop u}
{P_{u,\infty }}\left( {{R_u}} \right) \approx 1 - \sum\limits_{d = 0}^D {\frac{{{G_d}{{\left( {{\tau _d}} \right)}^\alpha }}}{{\Gamma \left( {\alpha  + 1} \right)}}{e^{ - \frac{{{\varepsilon _{u}}{\kappa ^2}{\beta ^2}\tau _d^2}}{{{\Omega _u}}}}}}.
\end{align}
\end{corollary}

\begin{proof} \label{Proof of corollary1}
\emph{
Under the condition $\rho  \to \infty $, the SINR for LU to decode data signal can be recast as
$\gamma _{uu}^\infty  = {\left| {{h_u}} \right|^2}/{\left( {\sum\nolimits_{m = 1}^M {\left| {h_{ru}^mh_{br}^m} \right|} } \right)^2}$. Hence, the SOP expression is shown as
\begin{align}\label{SINR uu proof}
{P_{u,\infty }}\left( {{R_u}} \right) &= {P_r}\left( {{C_u} < {R_u}} \right) \notag \\ &= {P_r}\left[ {\gamma _{uu}^\infty  < {2^{{R_u}}}\left( {1 + \gamma _{eu} } \right) - 1} \right] \notag \\
&= \int_0^\infty  {{f_Y}\left( y \right){F_X}} \left( {{\varepsilon _u}{\kappa ^2}{y^2}} \right)dy \notag \\
&= 1 - \int_0^\infty  {\frac{{{y^\alpha }{e^{ - \frac{y}{\beta }}}}}{{{\beta ^{\alpha  + 1}}\Gamma \left( {\alpha  + 1} \right)}}} {e^{ - \frac{{{\varepsilon _u}{\kappa ^2}{y^2}}}{{{\Omega _u}}}}}dy,
\end{align}where $X = {\left| {{h_u}} \right|^2}$, $Y = {\left| {\sum\nolimits_{m = 1}^M {\left| {h_{ru}^mh_{sr}^m} \right|} } \right|}$. Upon adopting Gauss-Laguerre approximation, the integral above can be resolved by replacing $y/{\beta}$ with $t$ and (33) is obtained. The proof is completed.
}
\end{proof}
\begin{remark}\label{remark1}
By substituting (\ref{asy sop u}) into (\ref{div define}), the secrecy diversity order for LU to decode data signal is equal to zero. This is due to the fact that the backscatter signal is regarded as interference when LU carries out SIC process and thus results in the appearance of error floor.
\end{remark}
\begin{corollary}\label{Corollary2}
Conditioned on $\rho  \to \infty $, the asymptotic SOP of LU to decode backscatter signal with ipSIC in RIS-AmBC networks can be given by
\begin{align}\label{asy sop c ipSIC}
P_{c,\infty }^{ipSIC}\left( {{R_c}} \right) \approx \sum\limits_{d = 0}^D {\frac{{{G_d}}}{{\Gamma \left( {\alpha  + 1} \right)}}\gamma \left( {\alpha  + 1,\sqrt {\frac{{{\varepsilon _{c3}}\varpi {\Omega _{ipu}}{\tau _d}}}{{{\beta ^2}}}} } \right)} ,
\end{align}where ${\varepsilon _{c3}} =  \left\{ {2^{{R_c}}}\left[ {1 +  \left({{\kappa ^2}{\rho _e}Q{\Omega _{sr}}{\Omega _{re}}}\right)/\left({\varpi {\rho _e}{\Omega _{ipe}}}\right)} \right] - 1 \right\}$ $/{\kappa ^2}$.
\end{corollary}

\begin{proof} \label{Proof of corollary2}
\emph{
Under the condition $\rho  \to \infty $, the SINR for LU to decode backscatter signal with ipSIC can be recast as
$\gamma _{uc,\infty }^{ipSIC} = {\kappa ^2}{\left| {\sum\nolimits_{m = 1}^M {\left| {h_{ru}^mh_{br}^m} \right|} } \right|^2}/\left( {\varpi {{\left| {{h_{ipu}}} \right|}^2}} \right)$. Hence, the SOP expression within high SNR region is shown as
\begin{align}\label{SINR uc ipsic proof}
P_{c,\infty }^{ipSIC}\left( {{R_c}} \right) &= {P_r}\left[ {\gamma _{uc,\infty }^{ipSIC} < {2^{{R_c}}}\left( {1 + \gamma _{ec}^{ipSIC}} \right) - 1} \right] \notag \\ &= \int_0^\infty  {{f_X}\left( x \right){F_Y}} \left( {\sqrt {{\varepsilon _{c3}}\varpi x} } \right)dx  \notag \\ &= \int_0^\infty  {\frac{{{e^{ - \frac{x}{{{\Omega _{ipu}}}}}}}}{{{\Omega _{ipu}}\Gamma \left( {\alpha  + 1} \right)}}} \gamma \left( {\alpha  + 1,\frac{1}{\beta }\sqrt {{\varepsilon _{c3}}\varpi x} } \right)dx.
\end{align}where $X = {\left| {{h_{ipu}}} \right|^2}$, $Y = {\left| {\sum\nolimits_{m = 1}^M {\left| {h_{ru}^mh_{sr}^m} \right|} } \right|}$. Referring to the procedure in (\ref{SINR uu proof}). We can acquire (\ref{SINR uc ipsic proof}) and the proof is completed.
}
\end{proof}
When $\varpi  = 0$, the asymptotic SOP of LU to decode backscatter signal with pSIC in RIS-AmBC networks when $\rho  \to \infty $ can be given by
\begin{align}\label{asy sop c pSIC}
P_{c,\infty }^{pSIC}\left( {{R_c}} \right) \approx {\left( {\frac{{\sqrt {{\varepsilon _{c2}}} }}{\beta }} \right)^{\alpha  + 1}}\sum\limits_{k = 0}^\infty  {\frac{{{{\left( {{\varepsilon _{c2}}} \right)}^{\frac{k}{2}}}{{e^{ - \frac{{\sqrt {{\varepsilon _{c2}}} }}{\beta }}}}}}{{\Gamma \left( {\alpha  + k + 2} \right)\beta }}}.
\end{align}

\begin{proof} \label{Proof of corollary2}
\emph{
Under the condition $\rho  \to \infty $, the CDF of SINR for LU to decode the backscatter signal with pSIC can be approximatively expressed as \emph{\cite{gautschi1979computational}}
\begin{align}\label{cdf SINR uc psic proof}
F_{\gamma _{uc}^{pSIC}}^\infty \left( x \right) = {\left( {\frac{1}{{\kappa \beta }}\sqrt {\frac{x}{\rho }} } \right)^{\alpha  + 1}}\sum\limits_{k = 0}^\infty  {\frac{{e^{ - \frac{1}{{\kappa \beta }}\sqrt {\frac{x}{\rho }} }}}{{\Gamma \left( {\alpha  + k + 2} \right)}}} {\left( {\frac{1}{{\kappa \beta }}\sqrt {\frac{x}{\rho }} } \right)^k}. \tag{38}
\end{align}
Hence, the asymptotic SOP expression for LU decode backscatter signal with pSIC is shown as
\begin{align}\label{sop uc psic proof}
P_{c,\infty }^{pSIC}\left( {{R_c}} \right) &= {P_r}\left[ {\gamma _{uc}^{pSIC} < {2^{{R_c}}}\left( {1 + \gamma _{ec}^{pSIC}} \right) - 1} \right] \notag \\  &= {F_{\gamma _{uc}^{pSIC}}^\infty}\left( {\rho {\kappa ^2}{\varepsilon _{c2}}} \right) \notag \\ &= {\left( {\frac{1}{\beta }\sqrt {{\varepsilon _{c2}}} } \right)^{\alpha  + 1}}\sum\limits_{k = 0}^\infty  {\frac{{{{\left( {{\varepsilon _{c2}}} \right)}^{\frac{k}{2}}}{e^{ - \frac{1}{\beta }\sqrt {{\varepsilon _{c2}}} }}}}{{\Gamma \left( {\alpha  + k + 2} \right){\beta ^k}}}} . \tag{39}
\end{align} The proof is completed.
}
\end{proof}
\begin{remark}\label{remark2}
By substituting (\ref{asy sop c ipSIC}) into (\ref{div define}), the secrecy diversity order for LU to decode backscatter signal with ipSIC is equal to zero. This is due to the fact that the residential interference caused by ipSIC can hinder LU's decoding process. By substituting (\ref{asy sop c pSIC}) into (\ref{div define}), the secrecy diversity order for LU to decode backscatter signal with pSIC is equal to $M{{{\pi ^2}} \mathord{\left/
 {\vphantom {{{\pi ^2}} {\left[ {2\left( {16 - {\pi ^2}} \right)} \right]}}} \right.
 \kern-\nulldelimiterspace} {\left[ {2\left( {16 - {\pi ^2}} \right)} \right]}}$.
\end{remark}
\subsection{Secrecy Throughput Analysis}
Apart from SOP-based reliability analyses, it is essential to estimate the secure effectiveness of RIS-AmBC networks. As a consequence, we define the secrecy throughput in delay-limited transmission mode of RIS-AmBC networks as follows
\begin{align}\label{SST define}
{T_\psi }\left( {{R_\zeta }} \right) = \left[ {1 - P_\zeta ^\psi \left( {{R_\zeta }} \right)} \right]{R_\zeta },
\end{align}where $\zeta  \in \left\{ {u,c} \right\}$ and $\psi  \in \left\{ \rm{{ipSIC,pSIC}} \right\}$. ${P_\zeta ^\psi \left( {{R_\zeta }} \right)}$ can be obtained from (\ref{SOP uu}), (\ref{SOP uc ipsic}) and (\ref{SOP uc psic}).

\subsection{Secrecy Energy Efficiency Analysis}
In RIS-AmBC networks, the total power consumption is composed of the transmitting power at BS and the hardware static dissipated power at BS, RIS as well as the user terminal \cite{2022XinweiYueRISNOMA,huang2019reconfigurable}, which is shown as
\begin{align}\label{total power consumption}
{P_{tot}} = {\vartheta ^{ - 1}}{P_s} + {P_u} + P_s^{hw} + P_{ris}^{hw} + P_{AmBC}^{hw},
\end{align}where $\vartheta $ denotes the efficiency of the transmitting power amplifier. ${P_u}$ is the hardware static power dissipated by the user terminal. $P_s^{hw}$ and $P_{ris}^{hw}$ represent the total hardware static dissipated power at the BS and RIS. $P_{AmBC}^{hw}$ is the power consumption the analog components of AmBC device, ie., the harvester as well as backscatter transceiver \cite{2018VanHuynhAmBC}. Furthermore, the secrecy energy efficiency of RIS-AmBC is defined as the ratio of secrecy throughput to the total power consumption, which is given by
\begin{align}\label{SEE define}
\eta_{ee}  = \frac{{{T_\psi }\left( {{R_\zeta }} \right)}}{{{P_{tot}}}},
\end{align}where $\zeta  \in \left\{ {u,c} \right\}$ and $\psi  \in \left\{ \rm{{ipSIC,pSIC}} \right\}$. ${{T_\psi }\left( {{R_\zeta }} \right)}$ and ${{P_{tot}}}$ can be acquired from (\ref{total power consumption}) and (\ref{SST define}), respectively.

\section{Numerical Results}\label{SectionIV}
\begin{table}[]
\renewcommand\arraystretch{1.5}
\centering
\caption{The table of Monte Carlo simulation parameters}
\setlength{\tabcolsep}{0.6mm}{		
\begin{tabular}{|l|l|}
\specialrule{0em}{1pt}{1pt}
\hline
Pathloss factor                       & $\lambda  = 2$                                                        \\ \hline
Frequency dependent factor             & $\eta  =  - 30$ dB                                                 \\ \hline
AWGN power                           & ${\sigma_u ^2} = \sigma _e^2 =  - 90$ dBm                                                        \\ \hline
Efficiency of transmitting power amplifier & $\vartheta = 0.32$                                                        \\ \hline
Residential interference                             & ${\Omega_{ipu}} = {\Omega_{ipe}} =  - 90$ dBm              \\ \hline
Target secrecy rate for data signal         & \begin{tabular}[c]{@{}l@{}}$R_u = 0.5$ BPCU\end{tabular}        \\ \hline
Target secrecy rate for backscatter signal         & \begin{tabular}[c]{@{}l@{}}$R_c = 0.1$ BPCU\end{tabular}        \\ \hline
Distance from the source to RIS/LU/Eve               & ${d_{sr}} = {d_{su}} = 20{\text{ m}}$, ${d_{se}} = 30{\text{ m}}$  \\ \hline
Distance from RIS to LU/Eve       & \begin{tabular}[c]{@{}l@{}}${d_{ru}} = 10{\text{ m}}$, ${d_{re}} = 20{\text{ m}}$\end{tabular} \\ \hline
Hardware static power consumption & \begin{tabular}[c]{@{}l@{}}$\{P_{u},P_{s}^{hw},P_{ris}^{hw},P_{AmBC}^{hw}\}$\\=$\{10,4,10,-31\}$ dBm\end{tabular} \\ \hline
\specialrule{0em}{1pt}{1pt}
\end{tabular}}
\end{table}

In this section, numerical results are presented to verify the derivations in Section \ref{SectionIII}. For the convenience of notation, the simulation parameters adopted are summarized in Table I 
\cite{2020WenjingBacandDire,2021JiakouRISBC}, where BPCU represents the abbreviation for bit per channel use. Noise power at LU and Eve is given by $\sigma _u^2 = \sigma _e^2 =  - 170 + 10\log \left( {{f_c}} \right) + {N_f} =  - 90$ dB, where the bandwidth is set to 10 MHz and noise figure is 10 dB \cite{ChaoZhang2022STARNOMA}. To highlight the secrecy performance of RISAmBC networks, the conventional AmBC networks without RIS and jammer-based AmBC networks are both taken into account as benchmarks. Note that the jammer-based model and simulation setups are in accordance with \cite{2023ShaoboAmBCSmartTransport}, where the RIS is replaced with a traditional backscatter device and a cooperative jammer is settled in the cell radiating artificial noise and confusing the Eve.

Fig. \ref{fig_2} plots the SOP versus transmitting power in RIS-AmBC networks, where $P=2$, $\rho_e=20$ dBm and $\kappa=0.5$. The theoretical analyses and simulations are well matched. On one hand, we can see that the SOP of the backscatter signal is much lower compared to conventional AmBC networks. This is because RIS can enhance the cascaded channel condition to improve the reliability of wireless propagation. On the other hand, as the power of backscatter signal increases, the SINR of LU decoding its own data signal decreases according to (\ref{SINR uu}), so that the SOP of the data signal in the traditional AmBC networks is lower than that of the data signal in RIS-AmBC networks. In addition, it is can be observed that when the transmission power is sufficient, the jammer-based secure transmission scheme outperforms the proposed scheme. This is because both the BS and the jammer have enough power to transmit data signals and interfere with Eve's eavesdropping behavior, without worrying about the negative impact of power being shared with each other. However, thanks to the cascaded link gain brought by RIS, the SOP of the backscatter signal in the RIS-AmBC network is significantly lower than that in the jammer-based AmBC network.

From the perspective of systemic SOP, it can be seen from Fig. \ref{fig_2} that the overall outage behaviour of RIS-AmBC networks is better than that of the AmBC networks without RIS assistance. We can also observe that with increasing the number of RIS elements, the secrecy performance of backscatter signal is promoted while that of the data signal is degraded. The reason for this phenomenon is that more RIS elements provide greater reflecting channel gains for backscatter communication links. Based on this, the backscatter signal is regarded as interference by LU when it detects the data signal. Hence, increasing RIS elements is beneficial to improve the received SINR of backscatter signal, but weakens that of data signal.  In addition, the error floors appear in the curves of data and backscatter signals with ipSIC demonstrated in \textbf{Remark \ref{remark1}} and \textbf{Remark \ref{remark2}}.

Fig. \ref{fig_3} plots the SOP versus the number of RIS elements, where $P=2$, $R_u=R_c=0.5$ BPCU, $P_s=30$ dBm, $\rho_e=20$ dBm and $\kappa=0.3$. It can be seen from the figure that with the rising value of $\kappa$, the SOP of the backscatter signal keeps declining. This is because more power is absorbed to convey the backscatter signal and the reception quality of backscatter signal is improved. As the number of RIS elements increases, the system SOP of RIS-AmBC networks first falls and then rises. This is due to the fact that the backscatter signal power is weak at first. By increasing the number of RIS elements, the received power of backscatter signal is gradually strengthened. However, the LU's ability to decode the data signal is greatly impaired when the power of the backscattered signal is too large, which in turn affects the system SOP performance.

\begin{figure}[t!]
    \begin{center}
        \includegraphics[width=3.2in,  height=2.5in]{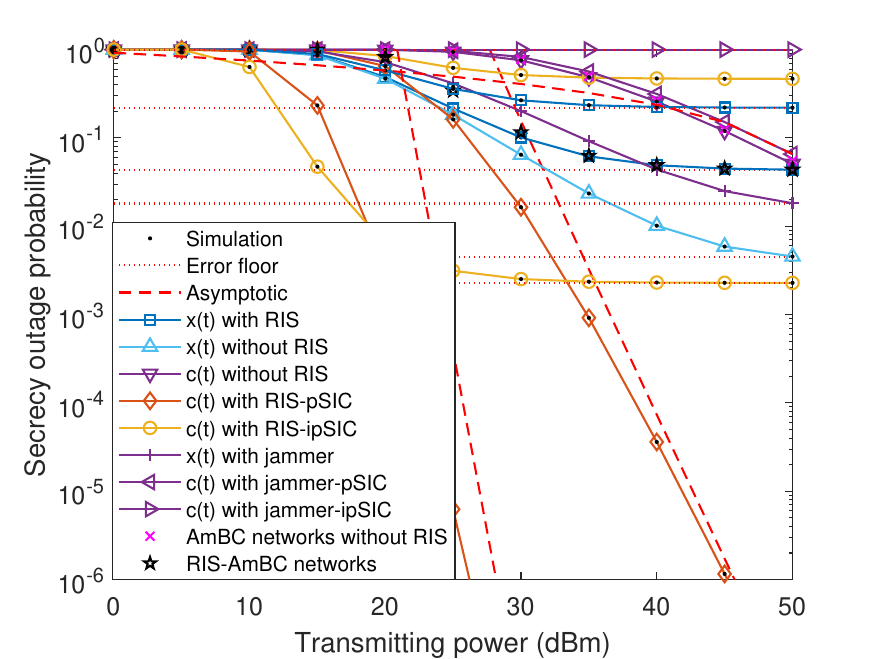}
        \caption{The SOP versus transmitting power in RIS-AmBC networks.
        }
        \label{fig_2}
    \end{center}
\end{figure}

\begin{figure}[t!]
    \begin{center}
        \includegraphics[width=3.2in,  height=2.5in]{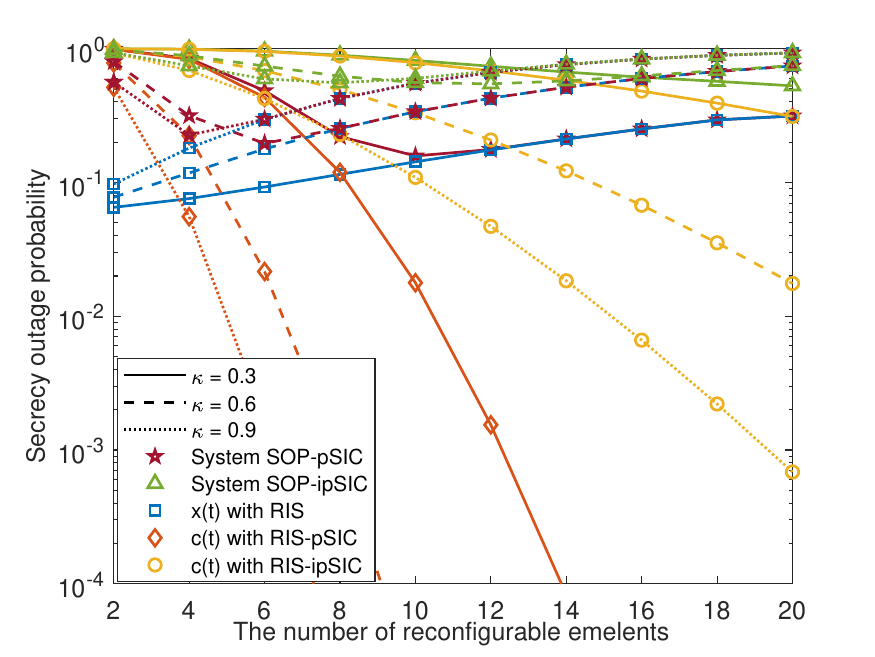}
        \caption{The SOP versus the number of RIS elements in RIS-AmBC networks.
        }
        \label{fig_3}
    \end{center}
\end{figure}

\begin{figure}[ht]
\centering
\subfigure[Backscatter signal]{
\begin{minipage}[t]{0.5\linewidth} 
\centering
\includegraphics[width=1.1\textwidth,height=1\textwidth]{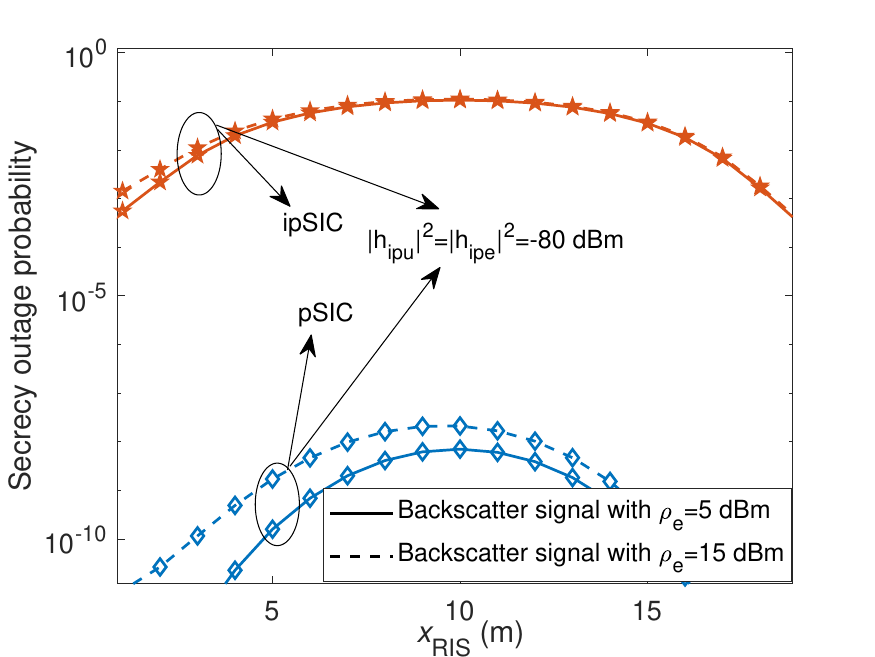}
\end{minipage}
}%
\subfigure[RIS-AmBC networks]{
\begin{minipage}[t]{0.5\linewidth} 
\centering
\includegraphics[width=1.1\textwidth,height=1\textwidth]{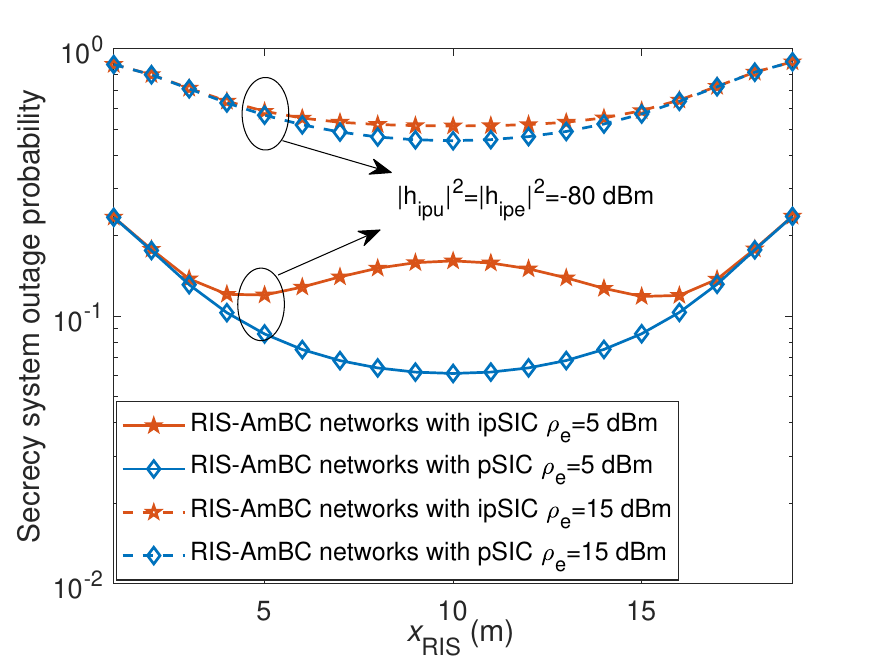}
\end{minipage}
}%

\centering
\caption{The SOP versus the horizontal distance of RIS in RIS-AmBC networks.}
\label{SOP_distance}
\end{figure}

Fig. \ref{SOP_distance} plots the SOP versus the horizontal distance of the RIS in RIS-AmBC networks, where $M=12$, $P=2$, $P_s=20$ dBm and $\kappa=0.5$. For purposes of illustration, the locations of the BS, LU, Eve and the RIS are labelled as (0,0), (20 m,0), (30 m,0) and ($x_{RIS}$,2) respectively in Euclidean coordinate space, where $x_{RIS}$ denotes the horizontal height of the RIS and varies from 1 m and 20 m. From Fig. (a) we can observe that the SOP of the backscatter signal is lower when the RIS is deployed close to the BS/LU. This is due to the fact that when the RIS is in the middle between the BS and the user, the incident and reflecting signals suffer from strong multiplicative fading, which impairs the secrecy rate of the backscatter signal. Furthermore, it can be seen from Fig. (b) that the optimal SOP for RIS-AmBC networks under pSIC can be achieved at the middle position ($x_{RIS}=10$ m), regardless of the wiretapping capability of Eve. This is due to the fact that the SOP of the backscatter signal with pSIC is much lower than that of the data signal, which means that the SOP of RIS-AmBC networks depends mainly on the data signal.  According to Fig. \ref{fig_2}, the data signal and the backscatter signal are just offset from each other, which results in the RIS being best deployed in the center for RIS-AmBC networks with pSIC. However, the SOP of the backscatter signal with ipSIC increases significantly and together with the data signal affects the overall performance of RIS-AmBC networks. In this case, the RIS can be deployed slightly closer to the BS or LU (e.g. $x_{RIS}=4, 14$ m).

\begin{figure}[t!]
    \begin{center}
        \includegraphics[width=3.2in,  height=2.5in]{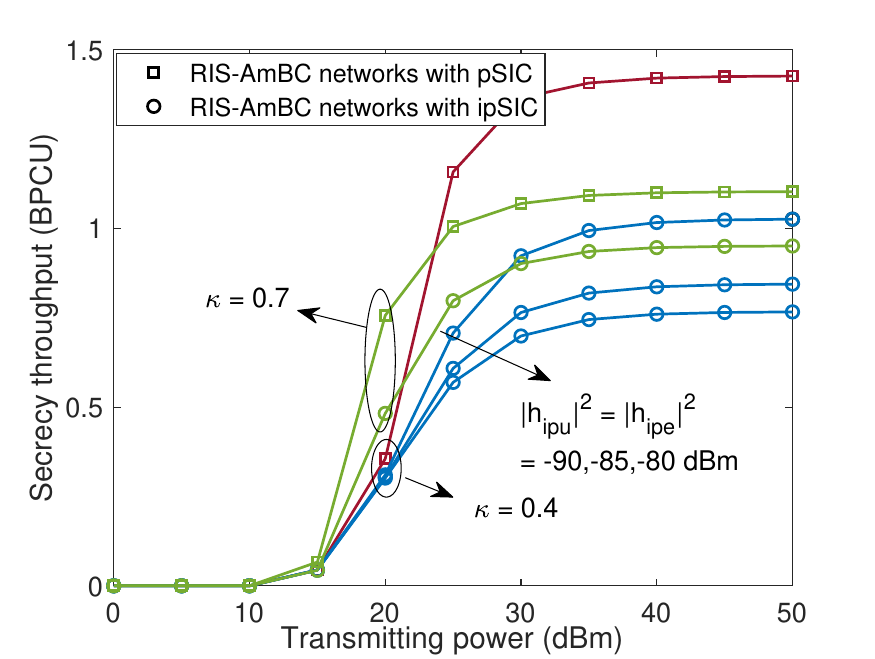}
        \caption{The secrecy throughput versus transmitting power in RIS-AmBC networks.
        }
        \label{SST SNR}
    \end{center}
\end{figure}

Fig. \ref{SST SNR} plots the secrecy throughput versus transmitting power, where $M=12$ $P=2$, $R_u=1$, $R_c=0.7$ BPCU, $\rho_e=20$ dBm and $\kappa=0.3$. It is observed that the secrecy throughput of RIS-AmBC networks with pSIC/ipSIC can not reach the theoretical upper limit (1.7 BPCU). This is because that the backscatter signal and residual interference can both corrupt LU's decoding process and result in error floors, so the SOP of data and backscatter signal will not be zero even with high SNRs. Furthermore, we can also observe that the secrecy throughput is impaired when the level of ipSIC increases. The reason for this is that the LU's received SINR is reduced and the resistance to eavesdropping is weakened.

\begin{figure}[t!]
    \begin{center}
        \includegraphics[width=3.2in,  height=2.5in]{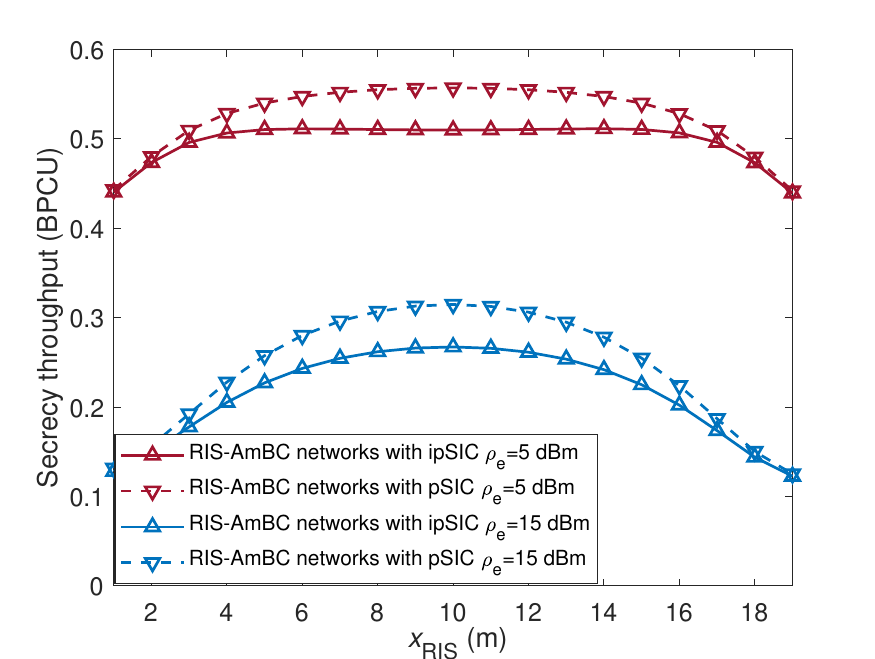}
        \caption{The secrecy throughput versus the horizontal distance in RIS-AmBC networks.
        }
        \label{SST distance}
    \end{center}
\end{figure}

Fig. \ref{SST distance} plots the secrecy throughput versus the horizontal distance of the RIS in RIS-AmBC networks, where $M=12$, $P=2$, $P_s=20$ dBm and $\kappa=0.5$. On one hand, it can be seen from the figure that regardless of the level of Eve's SNR, deploying the RIS midway between the BS and the LU is the optimal strategy to obtain the maximum secrecy throughput of RIS-AmBC networks. This can be explained by the fact that the target rate of the data signal occupies a larger share of the throughput and the SOP of the data signal is lower when the RIS is located near $x = 10$ m. On the other hand, considering the analyses of the relationship between SOP and RIS deployment locations in Fig. \ref{SOP_distance}, the RIS can also be deployed slightly closer to the BS or LU (e.g. $x_{RIS}=6, 12$ m) with a little acceptable loss of secrecy throughput under the practical ipSIC scenarios. In addition, we can observe that residual interference caused by ipSIC attenuates the SINR of the LU when decoding the backscatter signal, thus compromising the safe throughput of RIS-AmBC networks.

\begin{figure}[t!]
    \begin{center}
        \includegraphics[width=3.2in,  height=2.5in]{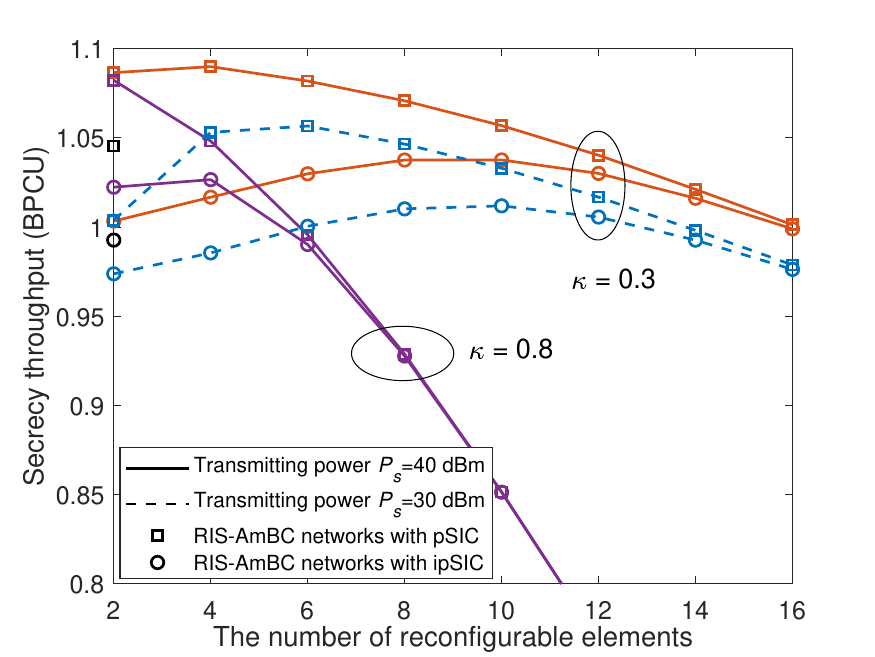}
        \caption{The secrecy throughput versus the number of RIS elements in RIS-AmBC networks.
        }
        \label{SST_M}
    \end{center}
\end{figure}

Fig. \ref{SST_M} plots the secrecy throughput versus the number of RIS elements in RIS-AmBC networks, where $P=2$, $P_s=40$ dBm, $R_u=1$ and $R_c=0.1$. As can be seen from the figure that increasing the reflecting coefficient $\kappa$ reduces the secrecy throughput of RIS-AmBC networks. As $\kappa$ increases, the amplification of the backscatter signal occurs, consequently compromising the SINR for LU during data signal decoding. Additionally, the system throughput is predominantly impacted by the data signal rate. Moreover, the observed trend indicates an initial rise and subsequent decline in the secrecy throughput of RIS-AmBC networks concerning variations in the number of RIS elements. The reason for this is that there exists an optimal value of \emph{M} that minimises the SOP according to Fig. \ref{fig_3} and this in turn will also lead to the existence of a value of \emph{M} that maximises secrecy throughput in delay-limited transmission mode.

\begin{figure}[t!]
    \begin{center}
        \includegraphics[width=3.2in,  height=2.5in]{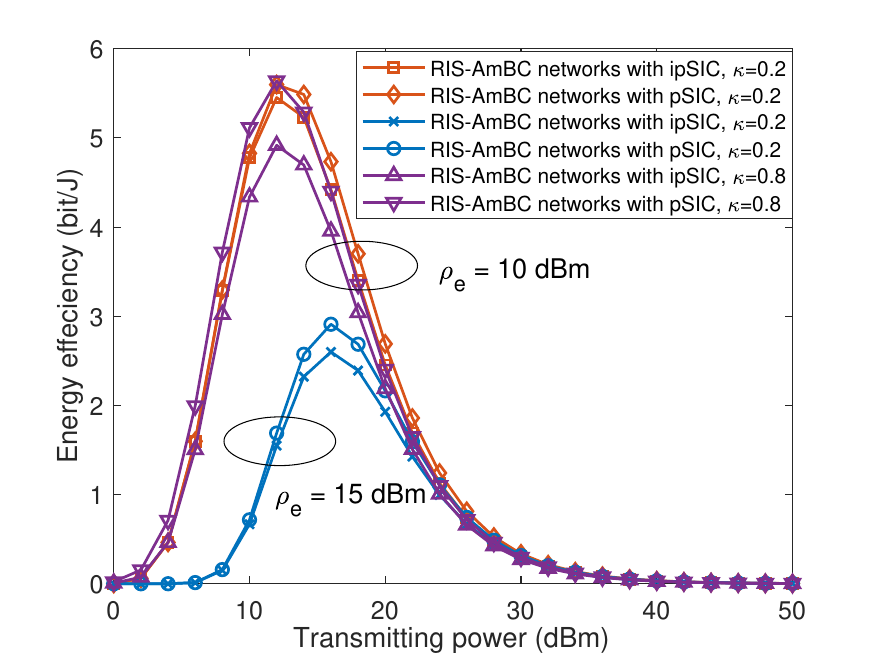}
        \caption{The secrecy energy efficiency versus transmitting power in RIS-AmBC networks.
        }
        \label{EneEffi_SNR}
    \end{center}
\end{figure}

Fig. \ref{EneEffi_SNR} plots the secrecy energy efficiency versus the transmitting power in RIS-AmBC networks., where $M=4$, $P=2$, $P_s=40$ dBm, $\rho_e=10$ dBm, $\kappa=0.3$, $R_u=1$ and $R_c=0.1$ BPCU. As the transmitting power rises, the secrecy energy efficiency gradually decreases after reaching a peak. This is due to the fact that according to Fig. \ref{SST SNR}, the secrecy throughput of RIS-AmBC networks tends to converge in the high transmitting power region, i.e. the numerator in (\ref{SEE define}) has an upper bound, while the transmitting power in the denominator part continues to increase, so the secrecy energy efficiency drops after reaching a maximum. In addition, when Eve has a high level of wiretapping ($\rho_e=15$ dBm), RIS-AmBC networks require a larger transmitting power to keep the secrecy rate stable in order to reach the optimal security energy efficiency.

\section{Conclusion}\label{Conclusion}\label{SectionV}
In this paper, the secrecy performance of RIS-AmBC networks has been investigated in detail. We have derived the closed-form and asymptotic expressions of SOP by taking into account both ipSIC and pSIC cases. The secrecy diversity orders of LU have been obtained in high SNR regions, which is connected with the residential interference and the number of RIS elements. In addition, the secrecy throughput and secrecy energy efficiency of RIS-AmBC networks are also analyzed. Numerical results demonstrated that RIS-AmBC can achieve superior secrecy outage behaviours compared to the conventional AmBC networks. To ensure the reliability and effectiveness of data and backscatter signal transmission, the best position to deploy RIS is not right in the middle of the BS and the LU, but slightly toward the BS/LU.

\appendices
\section*{Appendix~A} \label{AppendixA}
\renewcommand{\theequation}{A.\arabic{equation}}
\setcounter{equation}{0}
According to \cite{2022XinweiYueRISNOMA}, the  PDF of $Z = {{{\left| {{\mathbf{v}}_p^H{{\mathbf{\Lambda }}_{re}}{{\mathbf{h}}_{sr}}} \right|}^2}}$ can be given as
\begin{small}
\begin{align}\label{app1}
{f_{Z}}\left( z \right) = \frac{{2{z^{\frac{{Q - 1}}{2}}}}}{{\Gamma \left( Q \right){{\left( {\sqrt {{\Omega _{sr}}{\Omega _{re}}} } \right)}^{Q + 1}}}}{K_{Q - 1}}\left( {\sqrt {\frac{4z}{{{\Omega _{sr}}{\Omega _{re}}}}} } \right).
\end{align}
\end{small}Hence, the expectation of $Z = {{{\left| {{\mathbf{v}}_p^H{{\mathbf{\Lambda }}_{re}}{{\mathbf{h}}_{sr}}} \right|}^2}}$ is given by
\begin{small}
\begin{align}\label{app2}
\mathbb{E}\left( Z \right) &= \int_0^\infty  {z{f_Z}\left( z \right)} dz \notag \\ &= \int_0^\infty  {\frac{{2{{\left[ {\Gamma \left( Q \right)} \right]}^{ - 1}}{z^{\frac{{Q + 1}}{2}}}}}{{{{\left( {\sqrt {{\Omega _{sr}}{\Omega _{re}}} } \right)}^{Q + 1}}}}} {K_{Q - 1}}\left( {\sqrt {\frac{{4z}}{{{\Omega _{sr}}{\Omega _{re}}}}} } \right)dz\notag \\  &\mathop = \limits^{\left( a \right)} \int_0^\infty  {\frac{{4{{\left[ {\Gamma \left( Q \right)} \right]}^{ - 1}}{x^{Q + 2}}}}{{{{\left( {\sqrt {{\Omega _{sr}}{\Omega _{re}}} } \right)}^{Q + 1}}}}} {K_{Q - 1}}\left( {\sqrt {\frac{{4{x^2}}}{{{\Omega _{sr}}{\Omega _{re}}}}} } \right)dx,
\end{align}
\end{small} where ${\left( a \right)}$ represents the operation $z = {x^2}$.
Referring to \cite[Eq. (6.561.16)]{gradvstejn2000table}, we can obtain $\mathbb{E}\left( {Z} \right) = Q{\Omega _{sr}}{\Omega _{re}}$. Therefore, the SOP expressions for the Eve to decode data signal can be recast as
\begin{small}
\begin{align}\label{app4}
{P_u}\left( {{R_u}} \right) &= {\rm{P_r}}\left( {{C_u} < {R_u}} \right)\notag \\ &= {P_r}\left[ {{\gamma _{uu}} < {2^{{R_u}}}\left( {1 + {\gamma _{eu}}} \right) - 1} \right]\notag \\ &=
\int_0^\infty  {{f_Y}\left( y \right){F_X}\left[ {{\varepsilon _{u}}\left( {{\kappa ^2}{y^2}\rho  + 1} \right)/\rho } \right]} dy \notag \\
&= 1 - \int_0^\infty  {{e^{ - \left[ {\frac{{{\varepsilon _u}\left( {{\kappa ^2}{y^2} + 1} \right)}}{{\rho {\Omega _u}}} + \frac{y}{\beta }} \right]}}\frac{{{y^\alpha }}}{{{\beta ^{\alpha  + 1}}\Gamma \left( {\alpha  + 1} \right)}}} dy
\end{align}
\end{small}where $X = {\left| {{h_u}} \right|^2}$, $Y = {\left| {\sum\nolimits_{m = 1}^M {\left| {h_{ru}^mh_{sr}^m} \right|} } \right|}$ and ${\varepsilon _{u}} = {2^{{R_u}}}\left[ {1 + \left( {{\rho _e}{\Omega _e}} \right)/\left( {{\kappa ^2}Q{\Omega _{sr}}{\Omega _{re}}{\rho _e} + 1} \right)} \right] - 1$. Assuming $y/\beta  = t$, (\ref{app4}) can be further recast as
\begin{small}
\begin{align}\label{app5}
{P_u}\left( {{R_u}} \right) = 1 - \int_0^\infty  {\frac{{{e^{ - t}}{t^\alpha }}}{{\Gamma \left( {\alpha  + 1} \right)}}{e^{ - \frac{{{\varepsilon _u}\left( {{\kappa ^2}{\beta ^2}{t^2} + 1} \right)}}{{\rho {\Omega _u}}}}}} dt.
\end{align}
\end{small} With the assistance of Gauss-Laguerre integration, (\ref{SOP uu}) can be attained which completes the proof.

\bibliographystyle{IEEEtran}
\bibliography{mybib_test2}

\end{document}